\documentclass[conference]{IEEEtran}
\IEEEoverridecommandlockouts
\usepackage{graphicx}
\usepackage{epstopdf}
\usepackage{cite}
\usepackage{amsmath,amssymb,amsfonts}
\usepackage{algorithmic}
\usepackage{graphicx}
\usepackage{textcomp}
\usepackage{xcolor}
\usepackage{comment}
\usepackage{amsthm} 

\newtheorem{lemma}{Lemma}
\usepackage{comment}
\usepackage{makecell} 
\usepackage{booktabs}
\usepackage{array}

\def\BibTeX{{\rm B\kern-.05em{\sc i\kern-.025em b}\kern-.08em
    T\kern-.1667em\lower.7ex\hbox{E}\kern-.125emX}}

\begin{document}

\title{RC-Gossip: Information Freshness in Clustered Networks with Rate-Changing Gossip\\
}
\author{Irtiza Hasan$\qquad$Ahmed Arafa\\Department of Electrical and Computer Engineering\\ University of North Carolina at Charlotte, NC 28223\\
\emph{ihasan@charlotte.edu}$\qquad$\emph{aarafa@charlotte.edu}
\thanks{This work was supported by the U.S. National Science Foundation under Grants CNS 21-14537 and ECCS 21-46099.}}

\maketitle

\begin{abstract}
A clustered gossip network is considered in which a source updates its information over time, and end-nodes, organized in clusters through clusterheads, are keeping track of it. The goal for the nodes is to remain as fresh as possible, i.e., have the same information as the source, which we assess by the long-term average binary freshness metric. We introduce a smart mechanism of information dissemination which we coin \emph{rate-changing gossip} (RC-Gossip). Its main idea is that gossiping is directed towards nodes that need it the most, and hence the rate of gossiping changes based on the number of fresh nodes in the network at a given time. While Stochastic Hybrid System (SHS) analysis has been the norm in studying freshness of gossip networks, we present an equivalent way to analyze freshness using a renewal-reward-based approach. Using that, we show that RC-gossip significantly increases freshness of nodes in different clustered networks, with optimal cluster sizes, compared to traditional gossiping techniques.
\end{abstract}

\section{Introduction}

Real-time status updating is becoming increasingly important with the rise of intelligent agile systems equipped with Artificial Intelligence (AI). For example, autonomous vehicles, which can be modeled as gossip networks, need to communicate efficiently and transmit time-sensitive information to other vehicles nearby and coordinate\cite{xie2024}. Timely dissemination of fresh information from the source vehicle to others is crucial for such  real-time status updating systems. Systems which can be modeled as gossip networks, in which information freshness is crucial for successful operation, include wireless sensor networks (WSN)\cite{kadota2018}, Unmanned Aerial Vehicles (UAVs) and Internet of Things (IoT) networks\cite{choi2021}. Motivated by the ubiquitous and emerging applications of gossip networks, we study how we can optimize their gossiping strategies to keep them in sync with the source of information so as to maximize freshness throughout the network nodes. 

Status updating systems are commonly analyzed using the Age-of-Information (AoI) metric \cite{kaul2012realtime}, which has been applied in a wide range of scenarios for studying timeliness \cite{yates2021age}. A newer variant of AoI for gossip networks, Version Age-of-Information (VAoI) \cite{yates2021gossip}, has been introduced to keep track of the version of the source information to quantify age. The key mathematical tool in analyzing VAoI is Stochastic Hybrid System (SHS) analysis. This motivated a series of studies using SHS to quantify VAoI, see, e.g., \cite{Kaswan2022, Mitra2022,Kaswan2023,Srivastava2023,Mitra2024,Srivastava2025, Irtiza2025}, and the survey in \cite{Kaswan2025}. The studies of particular interest to our work are the works in \cite{BFM, Bastopcu2022}. In \cite{BFM}, the authors use the Binary Freshness Metric (BFM) in gossip networks where instead of keeping track of the version of information, the metric keeps track of whether nodes are fresh or stale. This turns out to be an equivalently effective way to characterize freshness of networks. In \cite{Bastopcu2022}, the authors analyze information dissemination in a fully connected gossip network using the long-term average BFM without invoking the SHS tool. In all these studies, the set of gossiping neighbors is fixed. That is, the set of nodes which receive gossip updates and the rate of gossip updates is the same throughout the information dissemination process and depends on the total number of connected (neighboring) nodes, which in turn relies on the network topology. 

\begin{figure}[t]
    \centering
    \includegraphics[width=.8\columnwidth]{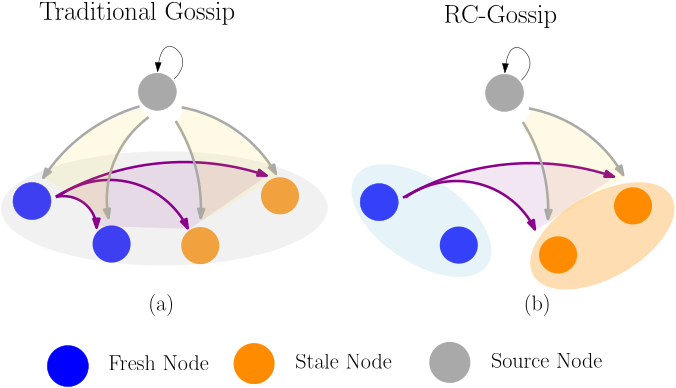}
    \caption{Compared to traditional gossip (left), RC-Gossip (right) directs gossiping towards stale nodes in the network, enhancing freshness by effectively changing gossiping rates.}
    \label{fig:RC Gossip}
    \vspace{-.2in}
\end{figure}
In this paper, we introduce the notion of {\it rate-changing gossip} (RC-Gossip), whose idea is based on directing gossiping towards stale nodes that need it the most, see Fig.~\ref{fig:RC Gossip}. Hence, the information dissemination (gossiping) rate becomes dependent on the number of fresh nodes in the network. Additionally, this makes gossiping {\it smart}, by saving unnecessary communications. 
On the other hand, it makes the analysis of such systems more challenging, since the topology is not fixed anymore; it changes according to which nodes are eligible to receive information, which further depends on their freshness status. Nevertheless, we overcome such difficulties by switching tools from the traditional SHS approach and, motivated by \cite{Bastopcu2022}, employ a renewal-reward (RR)-based approach to characterize the long-term average freshness in clustered networks under RC-Gossip. We first show that our RR approach provides equivalent results to those achieved by SHS approaches. Then, we use the RR tool to analyze the freshness improvement achieved under RC-Gossip. We do so for non-hierarchical networks first, and then use the results as building blocks to analyze hierarchical clustered networks. Finally, through numerical evaluations we show that there exists an optimal cluster size at which RC-Gossip significantly enhances freshness compared to traditional gossiping. 
\section{System Model} \label{sec:sys-mod}

We consider a system in which a source node updates itself according to a Poisson process of rate $\lambda_e$. The system consists of $n$ end-nodes that wish to follow the source's information. Nodes are divided into $m$ clusters of size $k$, i.e.,
\begin{align}
    m = \frac{n}{k},
\end{align}
with each cluster having a designated clusterhead (CH). Status updates from the source reach the end-nodes in two stages:
\begin{enumerate}
  \item \textbf{Source \(\to\) CHs:}
    The source sends updates to the \(m\) CHs according to a Poisson process of (total) rate \(\lambda_s\). 

  \item \textbf{CHs \(\to\) Nodes:}
    Each CH sends updates to its \(k\) nodes according to a Poisson process of (total) rate \(\lambda_c\). 
\end{enumerate}
Let $\nu_s(t)$ and $\nu_i(t)$ denote the version of the source information and node $i$ information at time $t$, respectively. We define the (binary) freshness of node $i$ at time $t$ as 
\begin{align}
    F_i(t)=\mathbb{I}_{\{\nu_i(t)=\nu_s(t)\}},
\end{align}
where $\mathbb{I}_{\{A\}}$ is the indicator function of event $A$. Our goal is to characterize the long-term average freshness of node $i$

\begin{align}
  \overline{F}_i \;\triangleq\;\limsup_{T\to\infty}\frac{1}{T}\mathbb{E}\left[\int_{0}^{T}F_i(t)\,dt\right].
\end{align}
for a given network topology.

We focus on two main network topologies in this work: a) disconnected (DC) network, and b) fully-connected (FC) network. For the DC network, nodes do not exchange updates among themselves; the source/CH sends an update, and once received, the information remains unchanged at the node until another update arrives. On the other hand, in the FC network all nodes gossip with each other. This means that if one node receives an update from the source/CH, it can propagate this update to other nodes. The latter occurs according to a Poisson process of (total) rate $\lambda_g$. In this work, CHs form a DC network, i.e., they can only get updates from the source. While end-nodes in each cluster can either form a DC or an FC network. 

We now illustrate our rate-changing gossip (RC-Gossip) strategy. In the RC-Gossip strategy, the source-to-CH, CH-to-node or node-to-node update rates change based on the number of stale nodes. For example, consider a non-hierarchical DC network. Initially, the average update rate from the source to a given node $i$ is given by $\lambda_s/n$ (by network symmetry). Now given that node $i$ is updated and is now fresh, the source excludes it from future updates and focuses on updating the rest of the nodes, and hence it effectively {\it increases its update rate} to $\lambda_s/(n-1)$. In general, when $k$ nodes are fresh, the source update rate to any of the remaining $n-k$ stale nodes is $\lambda_s/(n-k)$. Once the source updates itself, all nodes become stale and the update rates are reset to $\lambda_s/n$.

Now consider a non-hierarchical FC network with $k$ fresh nodes. In this case, under RC-Gossip any of the remaining $n-k$ nodes gets updates from the source with a rate of $\lambda_s/(n-k)$, and receives gossip updates from any of the $k$ fresh nodes with a rate of $\lambda_g/(n-k)$ as well. To put that into perspective, in traditional gossip we would have had these rate be fixed at $\lambda_s/n$ and $\lambda_g/(n-1)$, respectively, regardless of the number of fresh nodes. This means a lot of the connections in traditional gossip are sending out redundant information which can be optimized by allocating update rates to the more stale nodes in order to maximize the freshness of the network.

We note that freshness achieved in RC-Gossip serves as a theoretical upper bound on freshness achieved among networks which split the total gossip rate among the set of nodes. Fresh nodes in the network are propagating information targeting stale nodes, which is optimal. 

An illustrative depiction of RC-Gossip is shown in Fig.~\ref{fig:RC Gossip}.


\section{Renewal-Reward Analysis of Freshness in Gossip Networks}

\begin{figure}[t]
    \centering
    \includegraphics[width=.635\columnwidth]{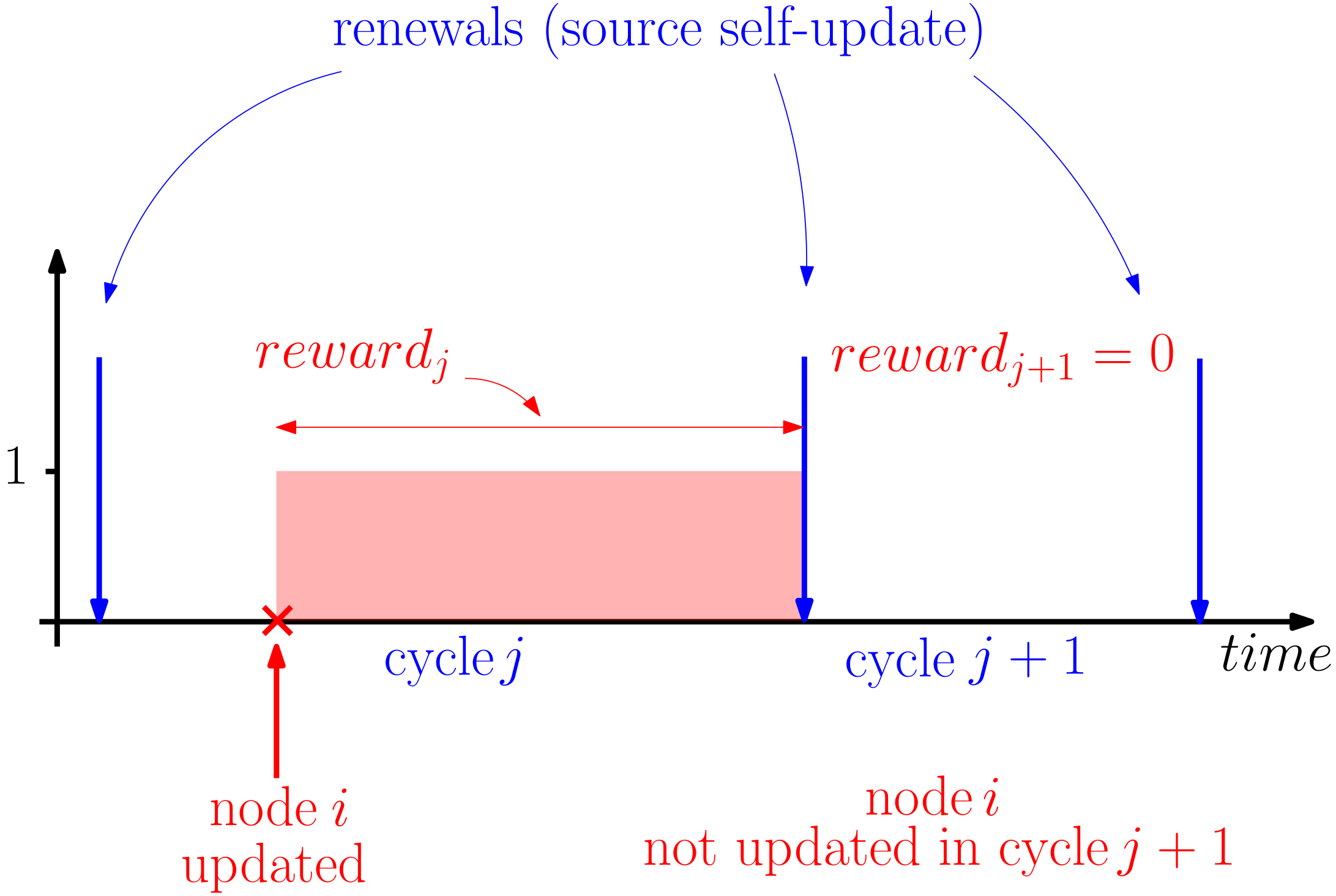}
    \caption{In cycle $j$ node $i$ gets updated; the reward (area) is given by $reward_j$. In cycle $j+1$, node $i$ does not get updated and thus $reward_{j+1} = 0$.}
    \label{fig:RR}
    \vspace{-.25in}
\end{figure}

In this section, we discuss the main tool we use to analyze RC-Gossip. Such tool, however, could be applied in general to analyze freshness for traditional gossip networks as well. The idea is to use the fact that, under binary freshness, the network status rests (renews) whenever the source updates itself. We denote by a {\it cycle} the time elapsed in between two consecutive source self-updates. Now observe that cycles are independent and identically distributed (i.i.d.) $\sim\exp(\lambda_e)$, and note that node $i$ has an average freshness given by the area under the binary freshness curve. During cycle $j$, such area is non-zero if and only if node $i$ receives an update, and is given by the duration of time during which it stays fresh, i.e., the time until the end of cycle $j$. Using the memoryless property of the exponential distribution, such time remaining is also $\sim\exp(\lambda_e)$. Since the network resets after each cycle, the probability of node $i$ receiving an update is the same for all cycles, which we denote by $p$. 

Finally, using the renewal-reward theorem, one can write the following:

\begin{align} \label{eq:RR}
    \overline{F}_i=\frac{(1-p)\cdot0+p\cdot1/\lambda_e}{1/\lambda_e}=p.
\end{align}

Thus, the expected long-term average freshness of a node in a gossip network is the same as the probability that the node gets updated before the source self-updates. The problem of calculating freshness of a node $i$ reduces to determining the probability $p$ that node $i$ is updated within a renewal cycle. The RR approach concepts are illustrated in Fig. \ref{fig:RR}.

In the next section, we use the aforementioned renewal-reward (RR)-based approach to analyze RC-Gossip in non-hierarchical settings, and then use the results as a building block to analyze RC-Gossip in clustered networks.
\section{Non-Hierarchical Networks} \label{sec:non-clustered}

We first show that the RR-based approach yields the same results as reported in \cite{BFM} using SHS.

\begin{lemma} \label{lemma:dc-norc}
    For a DC network with traditional gossip (denoted $DC_{\mathrm{noRC}}$), we have
    \begin{align} \label{eq_dcNoRC}
        \overline{F}_i^{DC_{\mathrm{noRC}}}=\frac{\lambda_s}{\lambda_s+n\lambda_e}.
    \end{align}
\end{lemma}

\begin{proof}
    In the $DC_{\mathrm{noRC}}$ setting, the source sends updates to each node at a fixed rate $\lambda_s/n$ throughout the entire cycle, independent of whether other nodes have been updated. Hence, the value of $p$ in \eqref{eq:RR} is simply given by the probability that an $\exp(\lambda_s/n)$ random variable is smaller than another, independent, $\exp(\lambda_e)$ random variable.
\end{proof}
\vspace{-0.03in}
Next, we present the first result for RC-Gossip.

\begin{lemma} \label{lemma:dc-src}
    For a DC network with RC-Gossip (denoted $DC_{\mathrm{RC}}$), we have
    \begin{align} \label{eq_dcRC}
        \overline{F}_i^{DC_{\mathrm{RC}}} = \frac{\lambda_s}{n \lambda_e} \left[1 - \left( \frac{\lambda_s}{\lambda_s + \lambda_e} \right)^n \right].
        \end{align}
\end{lemma}

\begin{proof}[Proof Sketch]
    In the $DC_{\mathrm{RC}}$ setting, when a node gets updated, the remaining stale nodes are assigned an increased update rate. Let $U_{i,m}$ denote the event that node $i$ is the $m$th node to be updated within a given cycle. Hence, we have
    \begin{align}
        \mathbb{P}(U_{i,1})=\frac{\lambda_s/n}{\lambda_s+\lambda_e}.
    \end{align}
    Next, we show that
    \begin{align}
        \mathbb{P}&(U_{i,m}) \nonumber \\
        &=\frac{(n-1)\lambda_s/n}{\lambda_s+\lambda_e}\cdot\frac{(n-2)\lambda_s/(n-1)}{\lambda_s+\lambda_e}\dots\frac{\lambda_s/(n-m+1)}{\lambda_s+\lambda_e}.
    \end{align}
    Finally, we note that in this case
    \begin{align}
        p=\sum_{m=1}^n\mathbb{P}(U_{i,m}).
    \end{align}
    Rearranging the above gives the result in the lemma.
\end{proof}

We now argue that the freshness in \eqref{eq_dcRC} is larger than that in \eqref{eq_dcNoRC} for $n\geq2$. Indeed, after some algebraic manipulations one can show that this is equivalent to showing that
\begin{align}\label{eq_dcRCgdcnoRC}
\left( 1 + \frac{\lambda_e}{\lambda_s} \right)^n > 1 + n \frac{\lambda_e}{\lambda_s},
\end{align}
which is readily given by the binomial expansion theorem.

We now turn our attention to FC networks. Observe that in this case RC-Gossip can vary: it can be applied at the source only (denoted $FC_{\mathrm{sRC}}$), or at the source and gossiping nodes as well (denoted $FC_{allRC}$). Clearly, we expect to have
\begin{align}
\overline{F}_i^{FC_{allRC}}\geq\overline{F}_i^{FC_{\mathrm{sRC}}}\geq\overline{F}_i^{FC_{\mathrm{noRC}}},
\end{align}
where $FC_{\mathrm{noRC}}$ denotes the traditional gossip FC network. Due to space limits, we only discuss the (best) $FC_{allRC}$ case.

\begin{lemma}
    For the $FC_{allRC}$ network, we have
    \begin{align}
        \overline{F}_i^{FC_{allRC}}=\frac{1}{n}\sum_{k=1}^{n} \prod_{j=1}^{k}\frac{\lambda_s+(j-1)\lambda_g}{\lambda_s+(j-1)\lambda_g+\lambda_e}.
        \label{eq:FC_allRC_ex}
    \end{align}
    
\end{lemma}

\begin{proof}[Proof Sketch]
    We first express the freshness of node $i$ as 
    \begin{align} \label{eq:FC_G}
        p=\sum_{k=1}^{n} \left[ \prod_{j=1}^{k-1} \tau_j \right] q_k,
    \end{align}
    where $q_k$ is the probability that node $i$ is the $k$th node to be updated after $k-1$ nodes have been updated already, and $\tau_j$ is the probability that the $j$th update goes to some other stale node other than node $i$ when $j-1$ nodes are fresh. We then show that
    \begin{align}
        q_k=&\frac{\lambda_s+(k-1)\lambda_g}{(n-k+1)\bigl[\lambda_s+(k-1)\lambda_g+\lambda_e\bigr]}, \\
        \tau_j=&\frac{(n-j)\bigl[\lambda_s+(j-1)\lambda_g\bigr]}{(n-j+1)\bigl[\lambda_s+(j-1)\lambda_g+\lambda_e\bigr]}.
    \end{align}
    Finally, the result is obtained upon substituting the above in \eqref{eq:FC_G} and performing some algebraic manipulations. 
\end{proof}
As a quick consistency check, we note that if $\lambda_g = 0$, the product in \eqref{eq:FC_allRC_ex} becomes $\bigl(\frac{\lambda_s}{\lambda_s + \lambda_e}\bigr)^k$, and the geometric sum gives the $DC_{\mathrm{RC}}$ result in Lemma~\ref{lemma:dc-src}.
\section{Clustered Networks} \label{sec:clustered}

We now use the previous section's results as a building block to analyze RC-Gossip in clustered networks. Let $p_{\mathrm{CH}}$ denote the probability that a CH is updated within the renewal cycle, and $p_{\mathrm{node}|\mathrm{CH}}$ denote the probability that an end-node is updated by its fresh CH within the same cycle. Thus, combining the two stages mentioned at the beginning of Section~\ref{sec:sys-mod}, we have
\vspace{-0.05in}
\begin{align}
    p=p_{\mathrm{CH}}\cdot p_{\mathrm{node}|\mathrm{CH}}.
\end{align}
Note that due to the memoryless property of the exponential distribution, once the CH is updated, the time remaining for the source self-update is still $\sim\exp(\lambda_e)$. The above observation simplifies the analysis of clustered networks via decomposition.

Observe that there are two main networks with varying employment of RC-Gossip: source-CH and CH-nodes. The source-CH network is either $DC_{\mathrm{noRC}}$ or $DC_{\mathrm{RC}}$, while the CH-nodes network could either be in one of the two DC network configurations, or be in the $FC_{\mathrm{noRC}}$, $FC_{\mathrm{sRC}}$ or $FC_{allRC}$ configuration. We present the final expressions for average end-node freshness under different network configurations in Table~\ref{table:Clustered-DC} for clustered DC networks, and in Table~\ref{table:Clustered-FC} for clustered FC networks. Since $FC_{\mathrm{sRC}}$  has a marginally higher freshness than $FC_{\mathrm{noRC}}$ especially for large $n$ (cf. Section~\ref{sec:num}), we consider $FC_{\mathrm{allRC}}$ to show the effectiveness of RC-Gossip in clustered FC networks. To highlight the increase in freshness due to RC-Gossip, we consider a case without RC-Gossip whether at source-CH or CH-nodes ($DC_{\mathrm{noRC}}, FC_{\mathrm{noRC}}$), a case with RC-Gossip only at source-CH ($DC_{\mathrm{RC}}, FC_{\mathrm{noRC}}$), and another case with RC-Gossip at both source-CH and CH-nodes ($DC_{\mathrm{RC}}, FC_{\mathrm{allRC}}$). We omit the proofs of these expressions due to space limits.

\vspace{-0.06in}

\begin{table}[!t]
\centering
\renewcommand{\arraystretch}{3.2} 
\caption{Average End-Node Freshness in Clustered DC Networks}
\label{table:disconnected}
\begin{tabular}{|m{3cm}|m{5cm}|}
\hline
\textbf{Network Config.} & \textbf{Average End-Node Freshness } \\ \hline
\((DC_{\mathrm{noRC}},DC_{\mathrm{noRC}})\) &
\raisebox{0.5\depth}{%
  \(\displaystyle
    \frac{\lambda_s}{\lambda_s + m\,\lambda_e} \times \frac{\lambda_c}{\lambda_c + k\,\lambda_e}
  \)
}
\\ \hline
\((DC_{\mathrm{noRC}},DC_{\mathrm{RC}})\) &
\raisebox{0.5\depth}{%
  \(\displaystyle 
    \frac{\lambda_s}{\lambda_s + m\,\lambda_e} \times 
    \frac{\lambda_c}{k\,\lambda_e}
    \left[1 - \left(\frac{\lambda_c}{\lambda_c + \lambda_e}\right)^k\right]
  \)
}
\\ \hline
\((DC_{\mathrm{RC}},DC_{\mathrm{noRC}})\) &
\raisebox{0.5\depth}{%
  \(\displaystyle 
    \frac{1}{m}\,\frac{\lambda_s}{\lambda_e}
    \left[1 - \left(\frac{\lambda_s}{\lambda_s + \lambda_e}\right)^m\right]
    \times \frac{\lambda_c}{\lambda_c + k\,\lambda_e}
  \)
}
 \\ \hline
\((DC_{\mathrm{RC}},DC_{\mathrm{RC}})\) & 
\(\begin{array}{c}
\displaystyle 
\left[\frac{1}{m}\,\frac{\lambda_s}{\lambda_e}(1 - a_s^m)\right] \times 
\left[\frac{\lambda_c}{k\,\lambda_e}(1 - a_c^k)\right]
\\[-0.5ex]
\text{where } a_s = \frac{\lambda_s}{\lambda_s + \lambda_e},\quad 
a_c = \frac{\lambda_c}{\lambda_c + \lambda_e}
\end{array}\)

\\  \hline
\end{tabular}
\label{table:Clustered-DC}
\vspace{-.1in}
\end{table}

\begin{table*}[!t]
\centering
\renewcommand{\arraystretch}{3.11}
\caption{Average End-Node Freshness in Clustered FC Networks} 
\label{table:FC_clustered_final}
\begin{tabular}{|p{2.5cm}|>{\centering\arraybackslash}p{14.5cm}|}
\hline
\textbf{Network Config.} & \textbf{Average End-Node Freshness} \\
\hline

\renewcommand{\arraystretch}{1.5}
\((DC_{\mathrm{noRC}},FC_{\mathrm{noRC}})\) &
\(\displaystyle 
\frac{\lambda_s}{\lambda_s + m\,\lambda_e} \!\cdot\!
\sum_{r=1}^{k}
\left( \prod_{i=1}^{r-1}
\frac{(k-i)\left(\frac{\lambda_c}{k} + (i-1)\frac{\lambda_g}{k-1}\right)}
     {(k-i+1)\left(\frac{\lambda_c}{k} + (i-1)\frac{\lambda_g}{k-1}\right) + \lambda_e}
\right)
\!\cdot\! 
\frac{\frac{\lambda_c}{k} + (r-1)\frac{\lambda_g}{k-1}}
     {(k-r+1)\left(\frac{\lambda_c}{k} + (r-1)\frac{\lambda_g}{k-1}\right) + \lambda_e}
\)
\\
\hline

\((DC_{\mathrm{RC}},FC_{\mathrm{noRC}})\) &
\(\displaystyle
\frac{1}{m}\,\frac{\lambda_s}{\lambda_e}
\left[1 - \left(\frac{\lambda_s}{\lambda_s + \lambda_e}\right)^m\right]
\!\cdot\!
\sum_{r=1}^{k}
\left( \prod_{i=1}^{r-1}
\frac{(k - i)\left(\frac{\lambda_c}{k} + (i - 1)\frac{\lambda_g}{k - 1}\right)}
     {(k - i + 1)\left(\frac{\lambda_c}{k} + (i - 1)\frac{\lambda_g}{k - 1}\right) + \lambda_e}
\right)
\!\cdot\!
\frac{\frac{\lambda_c}{k} + (r - 1)\frac{\lambda_g}{k - 1}}
     {(k - r + 1)\left(\frac{\lambda_c}{k} + (r - 1)\frac{\lambda_g}{k - 1}\right) + \lambda_e}
\)

\\ \hline

\((DC_{\mathrm{RC}},FC_{\mathrm{allRC}})\) &
\(\begin{array}{c}
\displaystyle
\frac{1}{m}\,\frac{\lambda_s}{\lambda_e}
\left[1 - \left(\frac{\lambda_s}{\lambda_s + \lambda_e}\right)^m\right]

\displaystyle
\!\cdot\!
\frac{1}{k}
\sum_{r=1}^{k}
\prod_{i=1}^{r}
\frac{\lambda_c + (i - 1)\lambda_g}
     {\lambda_c + (i - 1)\lambda_g + \lambda_e}
\end{array}\)

\\ \hline

\end{tabular}
\label{table:Clustered-FC}
\vspace{-.1in}
\end{table*}

\section{Numerical Results and Discussion} \label{sec:num}
In this section, we show some numerical examples to further illustrate the theoretical results of this paper. We start by the non-hierarchical setting in Section~\ref{sec:non-clustered}. In Fig.~\ref{fig:comparison}, we plot the long-term average freshness versus the total number of nodes $n$ for the five different networks considered. A crucial metric in performance comparison is the ratio $\alpha=\lambda_e/\lambda_s$ (the source self-update rate compared to the source-to-node update rate): lower (resp. higher) values $\alpha$ correspond to higher (resp. lower) freshness. For DC networks, we see that RC-gossip always enhances freshness for a given number of nodes. As for FC networks, we see that there is not much enhancement in freshness comparing $FC_{\mathrm{sRC}}$ and $FC_{\mathrm{noRC}}$, but once all nodes implement RC-gossip, we start to see a substantial increase in freshness for any given number of nodes. A common pattern for all network configurations is that for relatively high values of $\alpha$, freshness drops down to zero. This is due to the fact that no matter what scheme of gossiping we use, the long-term average freshness for a node in the network will deteriorate due the source's higher self update rate. However, freshness with RC-gossip will deteriorate relatively slower.  

\begin{figure}[!t]
    \centering
    
    \includegraphics[width=.65\columnwidth]{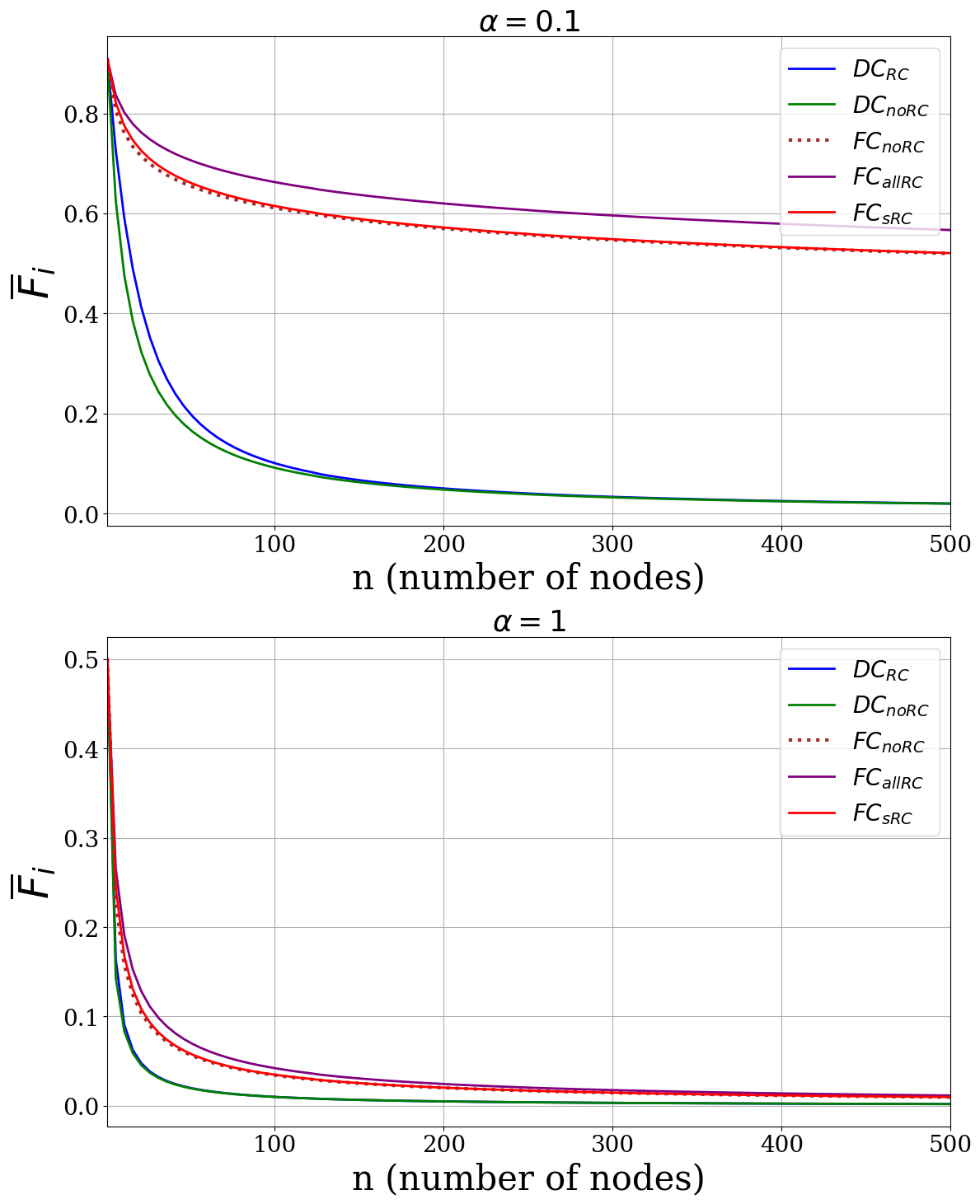}
    \caption{Comparison of $\overline{F}_i$ as a function of $n$ (number of nodes) for non-hierarchical networks under two different values of $\alpha = \lambda_e / \lambda_s \in \{0.1,\,1\}$. }
    \label{fig:comparison}
    \vspace{-.2in}
\end{figure}

Next, we move to clustered networks. In all cases here we fix $n=120$ nodes. Fig.~\ref{fig:Clustered DC} shows how the long-term average freshness for DC clustered networks vary with the cluster size $k$. We plot the four expressions shown in Table~\ref{table:Clustered-DC} under four different cases of update rates $(\lambda_e,\lambda_s,\lambda_c)$. In all cases, employing RC-gossip in both the source-CH and CH-nodes networks significantly improves the optimal average freshness achieved at the optimal cluster size. One notable observation is that employing RC-gossip only at the source-CH network or only at the CH-nodes network achieves the same optimal freshness, yet at different optimal cluster sizes, when $\lambda_s=\lambda_c$. However, when $\lambda_s>\lambda_c$ as in Case 2, RC-gossip favors being at the source-CH network, while the situation is reversed when $\lambda_s<\lambda_c$ as in Case 4. Basically, RC-gossip achieves higher gains when implemented at a source/CH with a relatively high update rate. 
Finally, in Fig.~\ref{fig:Clustered FC} we plot the expressions shown in Table~\ref{table:Clustered-FC} under varying cases of $(\lambda_e,\lambda_s,\lambda_c,\lambda_g)$. In all cases, as expected, $(DC_{\mathrm{RC}},FC_{\mathrm{allRC}})$ beats the other schemes significantly at the optimal cluster size.


\begin{figure}[!t]
    \centering
    \includegraphics[width=0.9\columnwidth]{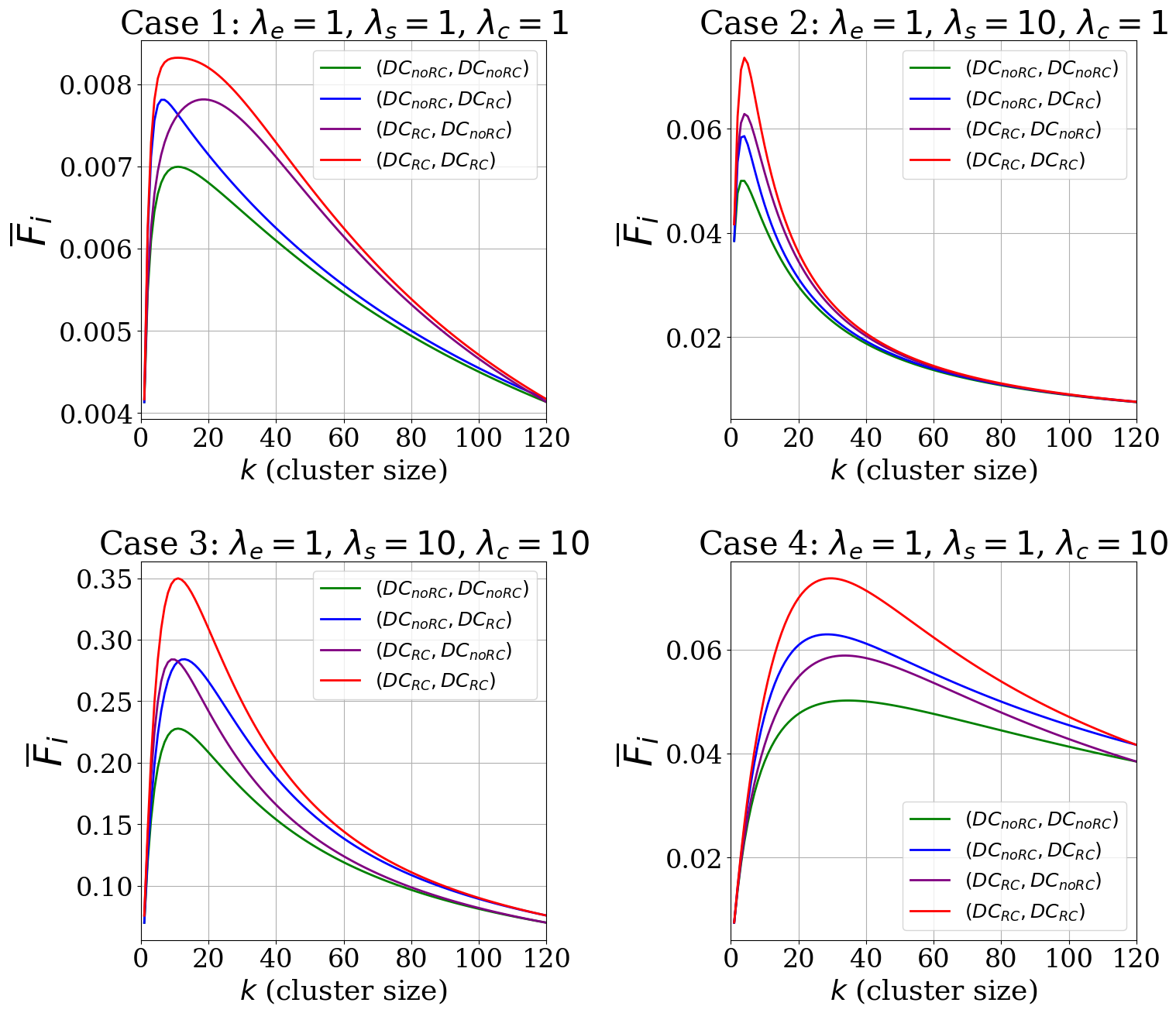}
    \caption{Comparison of $\overline{F}_i$ as a function of $k$ (cluster size) for DC clustered networks under different cases of $(\lambda_e,\lambda_s,\lambda_c)$.}
    \label{fig:Clustered DC}
    \vspace{-.2in}
\end{figure}

\begin{figure}[!t]
    \centering
    \includegraphics[width=0.9\columnwidth]{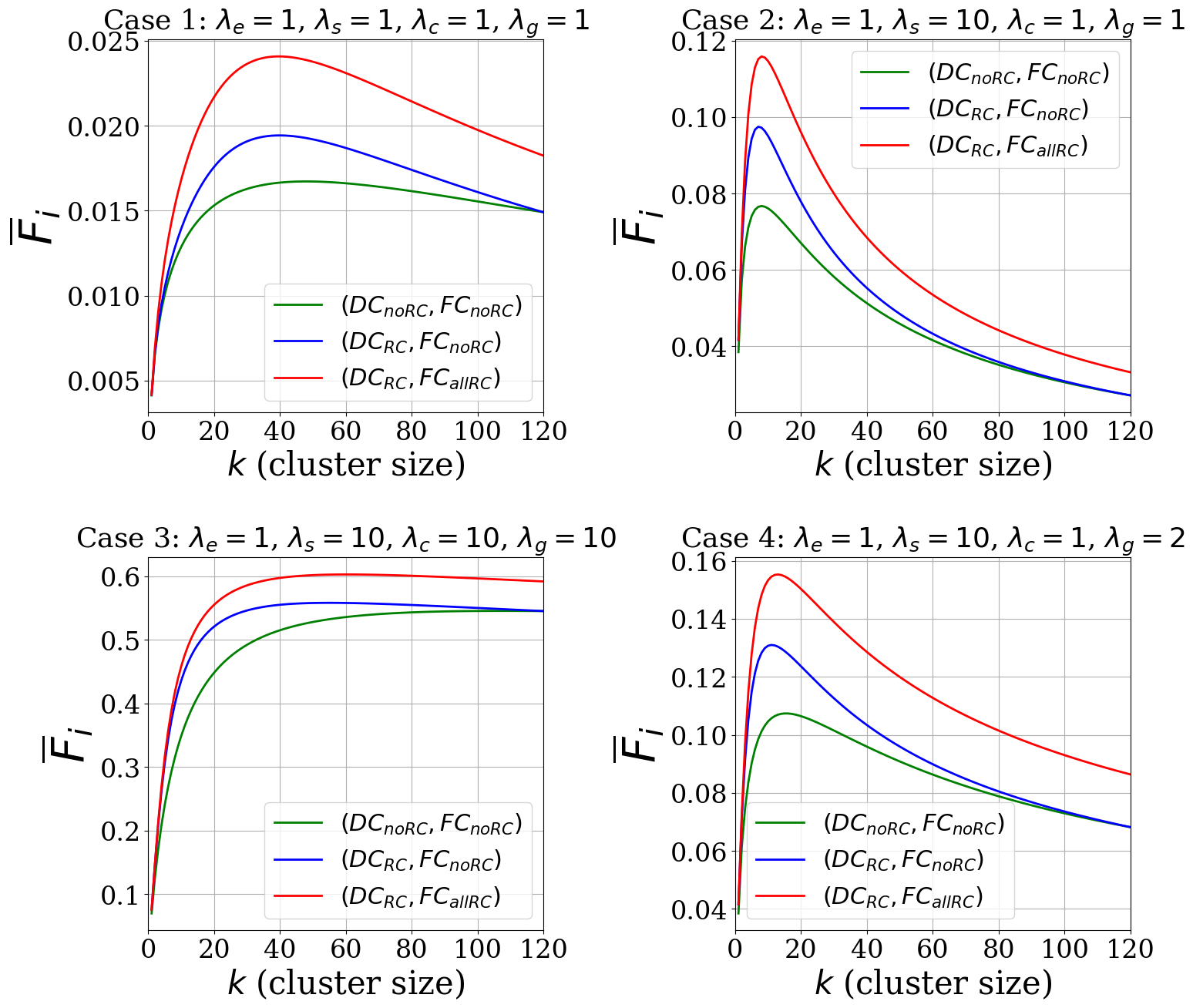}
    \caption{Comparison of $\overline{F}_i$ as a function of $k$ (cluster size) for FC clustered networks under different cases of $(\lambda_e,\lambda_s,\lambda_c,\lambda_g)$.}
    \label{fig:Clustered FC}
    \vspace{-.2in}
\end{figure} 

\section{Conclusion}
In this paper, we introduced and analyzed \textit{RC-Gossip} to improve information freshness in clustered networks. We analyzed freshness using a renewal-reward (RR)-based approach, providing a simpler alternative to the commonly employed SHS approach. Through theoretical analysis and numerical evaluations, we showed that RC-Gossip improves freshness across various gossip network topologies at optimal cluster sizes. Numerical studies further showed that optimal cluster sizes shift and freshness increases significantly with RC-Gossip compared to traditional gossiping schemes. Future work can explore extending to other network topologies and incorporating practical constraints such as limited energy resources or unreliable communication channels.


\end{document}